\documentclass[11pt]{article}
\usepackage{latexsym, amscd, amsfonts, mathrsfs, amsmath, amssymb, amsthm, stmaryrd, tikz-cd, mathrsfs, bbm, url, esint, todonotes, algorithm, fullpage, thm-restate, comment}
\usepackage[noend]{algpseudocode}

\usepackage[hidelinks]{hyperref}
\def\cA{\mathcal{A}}
\def\cB{\mathcal{B}}
\def\bE{\mathbb{E}}

\def\bP{\mathbb{P}}
\def\bZ{\mathbb{Z}}

\DeclareMathOperator\BSC{\mathrm{BSC}}

\DeclareMathOperator\Var{\mathrm{Var}}

\DeclareMathOperator{\polylog}{polylog}
\DeclareMathOperator\Bern{\mathrm{Bern}}

\newcommand{\rom}[1]{\textup{\uppercase\expandafter{\romannumeral#1}}}

\newtheorem{thm}{Theorem}
\newtheorem{lemma}[thm]{Lemma}

\newtheorem{coro}[thm]{Corollary}

\theoremstyle{definition}

\newtheorem{defn}[thm]{Definition}

\newtheorem{notn}[thm]{Notation}

\newtheorem{rmk}[thm]{Remark}

\newcommand{\NoisySorting}[1]{$\mathsf{NoisySorting}(#1)$}
\newcommand{\NoisyBinarySearch}[1]{$\mathsf{NoisyBinarySearch}(#1)$}
\newcommand{\GroupSorting}[2]{$\mathsf{GroupSorting}(#1, #2)$}
\newcommand{\OracleNoisyBinarySearch}[1]{$\mathsf{OracleNoisyBinarySearch}(#1)$}
\newcommand{\NoisyComparison}[2]{$\textsc{NoisyComparison}(#1, #2)$}

\begin{document}
\title{Optimal Bounds for Noisy Sorting}
\author{Yuzhou Gu \\ MIT \\ yuzhougu@mit.edu \and Yinzhan Xu \\ MIT \\ xyzhan@mit.edu}
\date{}

\setcounter{page}{0} \clearpage
\maketitle
\thispagestyle{empty}

\begin{abstract}
Sorting is a fundamental problem in computer science. In the classical setting, it is well-known that $(1\pm o(1)) n\log_2 n$ comparisons are both necessary and sufficient to sort a list of $n$ elements. In this paper, we study the Noisy Sorting problem, where each comparison result is flipped independently with probability $p$ for some fixed $p\in (0, \frac 12)$. As our main result, we show that
$$(1\pm o(1)) \left( \frac{1}{I(p)} + \frac{1}{(1-2p) \log_2 \left(\frac{1-p}p\right)} \right) n\log_2 n$$
noisy comparisons are both necessary and sufficient to sort $n$ elements with error probability $o(1)$ using noisy comparisons, where $I(p)=1 + p\log_2 p+(1-p)\log_2 (1-p)$ is capacity of BSC channel with crossover probability $p$. This simultaneously improves the previous best lower and upper bounds (Wang, Ghaddar and Wang, ISIT 2022) for this problem.

For the related Noisy Binary Search problem, we show that 
$$
    (1\pm o(1)) \left((1-\delta)\frac{\log_2(n)}{I(p)} + \frac{2 \log_2 \left(\frac 1\delta\right)}{(1-2p)\log_2\left(\frac {1-p}p\right)}\right)
$$
noisy comparisons are both necessary and sufficient to find the predecessor of an element among $n$ sorted elements with error probability $\delta$.
This extends the previous bounds of (Burnashev and Zigangirov, 1974), which are only tight for $\delta = 1/n^{o(1)}$.
\end{abstract}

\newpage

\section{Introduction}

Sorting is one of the most fundamental problems in computer science. Among the wide variety of  sorting algorithms, an important class is comparison-based sorting. It is well-known that sorting a list of $n$ numbers requires $(1-o(1))n\log n$ comparisons (see e.g., \cite{CLRS}), even if the input list is a random permutation.\footnote{All logarithms in the paper are of base $2$. } A good number of algorithms require only $(1+o(1))n\log n$ comparisons, such as binary insertion sort, merge sort (see e.g., \cite{knuth1997art}), and merge insertion sort \cite{ford1959tournament}. 

Sorting is run on many real-world applications with large data sets. %
Thus, it is important to consider faults that arise unavoidably in large systems. For comparison-based sorting, this  means that the result of a comparison could be incorrect. We consider one well-studied model for simulating such faults, studied in, e.g., \cite{horstein1963sequential,gal1978stochastic, pelc1989searching,berry1986discrete,feige1994computing,braverman2016parallel}. In this model, the result of each query (or noisy comparison in the context of sorting) is flipped independently with probability $p$ for some fixed $p\in (0, \frac 12)$ and repeated queries are allowed. Let us call it the noisy model. In this model, we define \NoisySorting{n} as the task to sort $n$ elements using noisy comparisons. 
As this model inherently has errors and no algorithm can always produce correct outputs,  we consider algorithms with $o(1)$ error probability. 

If we repeat each noisy comparison $\Theta(\log n)$ times, the majority of the returned results is the correct comparison result with high probability. Therefore, we can modify any classical comparison-based sorting algorithm with $O(n \log n)$ comparisons, by replacing each comparison with $\Theta(\log n)$ noisy comparisons, to get an $O(n \log^2 n)$-comparison algorithm for Noisy Sorting that succeeds with $1-o(1)$ probability. 

In fact, better algorithms were known. Feige, Raghavan, Peleg, and Upfal \cite{feige1994computing} gave an algorithm for \NoisySorting{n} with only $O(n \log n)$ noisy comparisons. Recently, Wang, Ghaddar, Wang \cite{wang2022noisy} analyzed the constant factor hidden in \cite{feige1994computing}'s algorithm,\footnote{The actual constant is quite complicated, see \cite{wang2022noisy} for details.} and gave an improved upper bound of $(1+o(1))  \frac{2}{\log\left(\frac{1}{\frac{1}{2}+\sqrt{p(1-p)}}\right)} n \log n$ noisy comparisons. 
They also showed that $(1-o(1))  \frac{1}{I(p)} n \log n$ comparisons are necessary for any algorithm with $o(1)$ error probability, where $I(p) = 1 - h(p)$ is the capacity of BSC channel with crossover probability $p$ (which is also the amount of information each noisy comparison gives) and $h(p) = -p \log p - (1-p) \log (1-p)$ is the binary entropy function. 

Despite these progresses, there is still a big gap between the upper and lower bounds. For instance, when $p = 0.1$, the constant in front of $n \log n$ for the lower bound is roughly $1.88322$, whereas the constant for the upper bound is roughly $6.21257$, more than three times as large. It is a fascinating question to narrow down this gap, or even close it like the case of classical sorting. 

As the main result of our paper, we show that it is indeed possible to close this gap by giving both an improved upper bound and an improved lower bound:
\begin{restatable}[Noisy Sorting Upper Bound]{thm}{SortUpperBound}
\label{thm:sort-upperbound}
There exists a  deterministic algorithm for \NoisySorting{n} with error probability $o(1)$ which always uses at most $$(1+o(1))\left(\frac{1}{I(p)} +\frac{1}{(1-2p) \log \frac {1-p}p}\right) n\log n$$ noisy comparisons.
\end{restatable}

\begin{restatable}[Noisy Sorting Lower Bound]{thm}{SortLowerBound}
\label{thm:sort-lowerbound}
Any (deterministic or randomized) algorithm for \NoisySorting{n} with error probability $o(1)$ must make at least $$(1-o(1)) \left(\frac 1{I(p)} + \frac 1{(1-2p) \log \frac{1-p}p}\right) n\log n$$ noisy comparisons in expectation, even if the input is a uniformly random permutation.
\end{restatable}

To compare  with previous works more quantitatively, our constant is roughly $2.27755$ for $p = 0.1$. When $p$ approaches $0$, i.e., when the noisy comparisons are almost errorless, our bounds approach $(1\pm o(1)) n \log n$, matching the bounds of classical sorting.

\paragraph{Noisy Binary Search. } We also study the  closely related Noisy Binary Search problem. In \NoisyBinarySearch{n}, we are given a sorted list of $n$ elements and an element $x$, and the goal is to find the predecessor of $x$ in the sorted list. Noisy Binary Search is well-studied \cite{horstein1963sequential, burnashev1974interval, feige1994computing, ben2008bayesian, pelc1989searching, dereniowski2021noisy} in the noisy model we consider, as well as in some other models~\cite{aslam1991searching, borgstrom1993comparison, dhagat1992playing, muthukrishnan1994optimal, rivest1980coping}.

Both previous algorithms \cite{feige1994computing,wang2022noisy} for \NoisySorting{n} use binary insertion sort at the high level. More specifically, the algorithms keep a sorted array of the first $i$ elements for $i$ from $1$ to $n$. When the $i$-th element needs to be inserted for $i > 1$, the algorithms use some implementations of \NoisyBinarySearch{i-1} to find the correct position to insert it, with error probability $o(1 / n)$ (it also suffices to have average error probability $o(1 / n)$). A simple union bound then implies that the error probability of the sorting algorithm is $o(1)$. 
The difference between the two algorithms is that~\cite{feige1994computing} uses an algorithm for Noisy Binary Search based on random walk, while \cite{wang2022noisy} builds on Burnashev and Zigangirov's algorithm \cite{burnashev1974interval}.

In fact, a variable-length version of Burnashev and Zigangirov's algorithm \cite{burnashev1974interval} achieves optimal query complexity for the Noisy Binary Search problem.\footnote{The original paper is in Russian. See \cite[Theorem 5]{wang2023noisy} for an English version.} They proved that \NoisyBinarySearch{n} with error probability $\delta$ can be solved using at most $\frac{\log n + \log (1/\delta) + \log((1-p)/p)}{ I(p) }$ noisy comparisons in expectation (we remark that their algorithm is randomized and assumes uniform prior). Information theoretical lower bound (e.g.,~\cite[Theorem B.1]{ben2008bayesian}) shows that $(1-\delta-o(1))\frac{\log n}{I(p)}$ noisy comparisons are necessary. These bounds essentially match when $1/n^{o(1)} \le \delta \le o(1)$. However, for the application of Noisy Sorting, we must have $\delta = o(1/ n)$ on average.
In this case, the lower and upper bounds do not match.

Dereniowski, {\L}ukasiewicz,  and Uzna{\'n}ski \cite{dereniowski2021noisy} designed alternative algorithms  Noisy Binary Search in various settings. They also considered a more general problem called  Graph Search.

We remark that Ben-Or and Hassidim \cite{ben2008bayesian}, likely unaware of \cite{burnashev1974interval}, claims an algorithm for solving \NoisyBinarySearch{n} with error probability $\delta$ using at most $(1-\delta-o(1))\frac{\log n + O(\log(1/\delta))}{ I(p) }$ noisy comparisons in expectation. However, as pointed out by \cite{dereniowski2021noisy}, there are potential issues in  \cite{ben2008bayesian}'s proof.

By using a similar technique as our lower bound for Noisy Sorting, we are also able to improve the lower bound for Noisy Binary Search for $\delta = 1/n^{\Omega(1)}$.
\begin{restatable}[Noisy Binary Search Lower Bound]{thm}{SearchLowerbound}
\label{thm:binarysearch-lowerbound}
Any (deterministic or randomized) algorithm for \NoisyBinarySearch{n} with error probability $\le \delta$ must make at least \begin{align*}
(1-o(1)) \left((1-\delta)\frac{\log n}{I(p)} + \frac{2 \log \frac 1\delta}{(1-2p)\log\frac {1-p}p}\right)
\end{align*}
noisy comparisons in expectation, even if the position of the element to search for is uniformly random.
\end{restatable}

This lower bound indicates that, any algorithm for \NoisySorting{n} purely based on binary insertion sort, which requires $n-1$ executions of Noisy Binary Search with average error probability $o(1/n)$,
needs at least $(1-o(1)) \left(\frac{1}{I(p)}+\frac{2}{(1-2p)\log \frac{1-p}{p}}\right) n \log n$ noisy comparisons in expectation to achieve error probability $o(1)$.
In light of Theorem~\ref{thm:sort-upperbound}, we see that any algorithm for Noisy Sorting purely based on binary insertion sort  cannot be optimal.

We also show an algorithm for \NoisyBinarySearch{n} that matches with the lower bound in Theorem~\ref{thm:binarysearch-lowerbound}:

\begin{restatable}[Noisy Binary Search Upper Bound]{thm}{SearchUpperBound}
\label{thm:binarysearch-upperbound}
There exists a randomized algorithm for \NoisyBinarySearch{n} with error probability $\le \delta$ that makes at most
\begin{align*}
(1+o(1)) \left((1-\delta)\frac{\log n}{I(p)} + \frac{2 \log \frac 1\delta}{(1-2p)\log\frac {1-p}p}\right)
\end{align*}
noisy comparisons in expectation.
\end{restatable}

\paragraph{Concurrent Work.} Concurrent work \cite{wang2023noisy} also improves over \cite{wang2022noisy} in many aspects. Our results are  stronger than their main results. Our algorithm in Theorem~\ref{thm:sort-upperbound} uses fewer noisy comparisons than their algorithms in expectation 
~\cite[Theorem 1 and 2]{wang2023noisy}. Our Theorem~\ref{thm:sort-lowerbound} is strictly stronger than \cite[Theorem 3]{wang2023noisy}. Our lower bound for insertion-based sorting algorithms (see discussion after Theorem~\ref{thm:binarysearch-lowerbound}) is strictly stronger than \cite[Theorem 4]{wang2023noisy}.

\paragraph{Other Related Works.}
People have also  considered Noisy Sorting in a slightly different noisy model, where every pair of elements is allowed to be compared at most once \cite{braverman2008noisy, klein2011tolerant,geissmann2017sorting, geissmann2020optimal, geissmann2019optimal}. Sorting is harder in this model than the model we consider, as it is  information-theoretically impossible to correctly sort all the elements with $1-o(1)$ probability for $p > 0$ \cite{geissmann2017sorting}. For $p<1/16$, Geissmann, Leucci, Liu, Penna~\cite{geissmann2019optimal} achieved an $O(n \log n)$ time algorithm that guarantees $O(\log n)$ maximum dislocation and $O(n)$ total dislocation with high probability, matching the lower bound given in \cite{geissmann2017sorting}. 

\paragraph{Acknowledgements.}
We thank Przemys{\l}aw Uzna{\'n}ski for telling us about the issue of \cite{ben2008bayesian}.

\subsection{Technical Overview}

Before we give the overview of our techniques, let us first give some intuitions about the constant, $\frac 1{I(p)} + \frac 1{(1-2p) \log \frac{1-p}p}$, in our bound. $I(p)$ is essentially the amount of information in bits each noisy query can give, and the ordering of $n$ elements requires $\log(n!) \approx n \log n$ bits to represent. Therefore, $\frac{n \log n}{I(p)}$ queries are necessary intuitively. On the other hand, $\frac {\log n}{(1-2p) \log \frac{1-p}p}$ is roughly the number of noisy comparisons required to compare two elements with error probability $\le \frac{1}{n}$. Therefore, even for the simpler task of checking whether a list of $n$ elements are sorted, roughly $\frac {n \log n}{(1-2p) \log \frac{1-p}p}$ queries seem necessary to compare all the adjacent elements with overall error probability $o(1)$. 
Therefore, the constants $\frac 1{I(p)}$ and $\frac 1{(1-2p) \log \frac{1-p}p}$ are natural. 

However, the above discussion only intuitively suggests a $\max\left\{\frac 1{I(p)}, \frac 1{(1-2p) \log \frac{1-p}p}\right\}$ lower bound on the constant, and a priori it is unclear how to strengthen it to their sum. To resolve this issue, we design an easier version of Noisy Sorting as an intermediate problem. In this version, suppose elements are split into continuous groups of sizes $\log n$ in the sorted order (and the groups are unknown to the algorithm). Any time the algorithm tries to compare two elements from the same group, the query result tells the algorithm that these two elements are in the same group, without noise; otherwise, the query result flips the result of the comparison with probability $p$, as usual. The algorithm is also not required to sort elements inside the same group correctly. We can show that in this version, each query still only gives roughly $I(p)$ bits of information, and the total amount of bits the output space can represent is $\log \frac{n!}{((\log n)!)^{n / \log n}} \approx n \log n$. Therefore, this simpler problem still requires approximately $\frac{n \log n}{I(p)}$ queries. By designing a reduction from this problem to Noisy Sorting, we can show that any algorithm for Noisy Sorting requires roughly $\frac{n \log n}{I(p)}$ queries comparing elements not from the same group. At the same time, inside each group, the algorithm needs to perform $\log n - 1$ comparisons with $\le \frac{1}{n}$ error probability, even for checking if all the groups are sorted. This part requires roughly $\frac {n \log n}{(1-2p) \log \frac{1-p}p}$ queries. 

Our algorithm that matches this lower bound can be  viewed as a remote variant of quicksort, with a large number of pivots and one level of recursion. First, we sample a random subset of elements $S$ of size $\frac{n}{\log n}$, which can be viewed as pivots. Sorting $S$ only contributes to lower order terms in the query complexity. These pivots separate the list to multiple sublists. 
Then, we use Noisy Binary Search   to find the correct sublist for the remaining elements, with error probability $\frac{1}{\log n}$. This step contributes to the $\frac{n \log n}{I(p)}$ part of the query complexity. Then we design an algorithm that correctly sorts $m$ elements with probability $\ge 1-o(\frac{1}{n})$ using only $O(m \log m) + (1+o(1)) \frac{m \log n}{(1-2p)\log \frac{1-p}{p}}$ queries. This algorithm is then used to sort all the sublists. Summing over all the sublists, the first $O(m \log m)$ term becomes $O(n \log \log n)$, as we expect each sublist to have size $\polylog(n)$. The second term sums up to $(1+o(1))\frac {n \log n}{(1-2p) \log \frac{1-p}p}$. The total error probability is $o(\frac{1}{n}) \cdot n = o(1)$ Finally, we have to correct the elements that are sent to the wrong sublists due to the error probability in the Noisy Binary Search algorithm, but we only expect to see at most $\frac{n}{\log n}$ of them, so the number of queries for handling them ends up being a lower order term. 

\section{Preliminaries}
In this section we introduce several basic notations, definitions, and lemmas used in the proof.

\subsection{Basic Notions}

\begin{notn}
We use the following notations throughout this paper.
\begin{itemize}
    \item For $k\in \bZ_{\ge 0}$, let $[k] := \{1,\ldots, k\}$.
    \item For two random variables $X, Y$, let $I(X;Y)$ denote the mutual information between $X$ and $Y$.
    \item For a distribution $P = (p_1,\ldots, p_k)$, define its entropy function %
    as $h(P) := -\sum_{i\in [k]} p_i \log p_i$.
    Specially, for $p\in [0, 1]$, define $h(p) := h(p,1-p)$. Define $I(p) := 1-h(p)$.
    \item For $p\in [0, 1]$, let $\Bern(p)$ denote the Bernoulli random variable, such that for $X\sim \Bern(p)$, we have $\bP[X=0]=1-p$, $\bP[X=1]=p$.
    \item For $p\in [0, 1]$, let $\BSC_p$ be the binary symmetric channel with crossover probability $p$, i.e., on input $x\in \{0,1\}$, the output $\BSC_p(x)$ follows distribution $\Bern(p+(1-2p)x)$.
    \item For two random variables $X, Y$, define $X\land Y := \min\{X, Y\}$.
    \item 
    Let $x$ be a variable and $f(x)$ and $g(x)$ be two functions which are defined for all large enough values of $x$.
    We say $f(x) = o_x(g(x))$ if $\lim_{x\to \infty} f(x)/g(x)=0$.
    Specially, if $x=n$, we omit the subscript and write $f(n) = o(g(n))$.
\end{itemize}
\end{notn}

We also assume some basic knowledge about information theory, such as chain rule of entropy and mutual information and Fano's inequality. See e.g.~\cite{polyanskiy2022information} for an overview. 

In Noisy Sorting and Noisy Binary Search, the algorithm is allowed to make  noisy comparison queries. 
\begin{defn}[Noisy Comparison Query]
Fix $p\in (0, \frac 12)$. Let $x, y$ be two comparable elements (i.e., either $x\le y$ or $y\le x$). We define the noisy comparison query
\NoisyComparison{x}{y} as returning $\BSC_p(\mathbbm{1}\{x< y\})$, where the randomness  is independent for every query.
\end{defn}
We will omit the crossover probability $p$ when it is clear from the context throughout the paper. 

With these definitions, we are ready to define the problems of study.
\begin{defn}[Noisy Sorting Problem]
Let \NoisySorting{n} be the following problem:
Given $n$ comparable elements $(a_i)_{i\in [n]}$, an algorithm is allowed to make query $\textsc{NoisyComparison}(a_i,a_j)$ for $i,j\in [n]$.
The goal is to output a permutation $(b_i)_{i\in [n]}$ of $(a_i)_{i\in [n]}$ such that
$b_i\le b_{i+1}$ for all $i\in [n-1]$.
\end{defn}

\begin{defn}[Noisy Binary Search Problem]
Let \NoisyBinarySearch{n} be the following problem:
Given $n$ elements $(a_i)_{i\in [n]}$ satisfying $a_i\le a_{i+1}$ for all $i\in [n-1]$ and an element $x$.
An algorithm is allowed to make query $\textsc{NoisyComparison}(s,t)$, where $\{s,t\}=\{a_i,x\}$ for some $i\in [n]$.
The goal is to output $\max (\{0\}\cup\{i\in [n]: a_i < x\})$ (i.e., the index of the largest element in $(a_i)_{i\in [n]}$ less than $x$).
\end{defn}

\begin{rmk}
By our definition of the Noisy Sorting Problem, we can WLOG assume that all elements $(a_i)_{i\in [n]}$ are distinct.
In fact, suppose we have an algorithm $\cA$ which works for the distinct elements case.
Then we can modify it such that if it is about to call $\textsc{NoisyComparison}(a_i,a_j)$, it instead calls $\textsc{NoisyComparison}(a_{\max\{i,j\}},a_{\min\{i,j\}})$ (but flip the returning bit if $i < j$).
Call the modified algorithm $\cA'$.
This way, to the view of the algorithm, the elements are ordered by $(a_i, i)$. 
It is easy to see that $\cA'$ has the same distribution of number of queries, and error probability is no larger than error probability of $\cA$. Thus, we get an algorithm which works for the case where elements are not necessarily distinct.

Similarly, by our definition of the Noisy Binary Search Problem, we can WLOG assume that all elements $(a_i)_{i\in [n]}$ and $x$ are distinct.
\end{rmk}

\subsection{Known and Folklore Algorithms}
\begin{thm}[{\cite[Corollary 1.6]{dereniowski2021noisy}}]
\label{thm:BHbinary}
There exists a randomized algorithm for \NoisyBinarySearch{n} with error probability at most $\delta$ using
$$ (1+o(1)) \left(\frac{\log n + O(\log (1/\delta))}{I(p)}\right)$$
noisy comparisons in expectation, and using $O(\log n)$ random bits always.
\end{thm}
\begin{rmk}
    \cite{dereniowski2021noisy} stated their problem as finding an element in a sorted array, but their algorithm can also find the predecessor. The only place their algorithm uses randomness is to shift the initial array randomly, which requires $O(\log n)$ random bits (as we can WLOG assume $n$ is a power of two in \NoisyBinarySearch{n}, it always takes $O(\log n)$ random bits to sample a random shift, instead of $O(\log n)$ random bits in expectation).

    We could also use the noisy binary search algorithm in \cite{burnashev1974interval} (see also \cite[Theorem 5]{wang2023noisy}), which achieves a slightly better bound on the number of noisy comparisons. However, the number of random bits used in Burnashev-Zigangirov's algorithm is $O(\log n)$ in expectation, not always. This would make our control of number of random bits used slightly more complicated.
\end{rmk}

\begin{coro}
\label{cor:simple_sort}
There exists a randomized algorithm for \NoisySorting{n} with failure probability $\delta$ using 
$O(n \log (n/\delta))$ noisy comparisons in expectation, and using $O(n\log n)$ random bits always.
\end{coro}
\begin{proof}
The algorithm keeps a sorted array of the first $i$ elements for $i$ from $1$ to $n$. Each time a new element needs to be inserted, the algorithm uses Theorem~\ref{thm:BHbinary} to find its correct position, with error probability $\delta / n$. By union bound, the overall error probability is bounded by $\delta / n \cdot n = \delta$, the expected number of queries is $n \cdot (1+o(1)) \left(\frac{\log n + \log(n/\delta)}{I(p)}\right) = O(\frac{n \log (n/\delta)}{I(p)})$, and the number of random bits used is $O(n\log n)$. 
\end{proof}

We show how to use repeated queries to boost the correctness probability of pairwise comparisons. This is a folklore result, but we present a full proof here for completeness.
\begin{lemma} \label{lem:repeated_query}
There exists a deterministic algorithm which compares two (unequal) elements using noisy comparisons with error probability $\le \delta$ using
\begin{align*}
\frac 1{1-2p} \left \lceil \frac{\log \frac{1-\delta}{\delta}}{\log \frac{1-p}p} \right\rceil = (1+o_{1/\delta}(1)) \frac{\log (1/\delta)}{(1-2p) \log((1-p)/p)}
\end{align*}
noisy queries in expectation. 
\end{lemma}
\begin{proof}

Consider Algorithm~\ref{algo:less_than} which maintains the posterior probability that $x < y$.

\begin{algorithm}[h]
\caption{}
\label{algo:less_than}
\begin{algorithmic}[1]
\Procedure{LessThan}{$x,y,\delta$}
\State $a\gets \frac 12$
\While{true}
\If{$\textsc{NoisyComparison}(x,y)=1$}
\State $a\gets \frac{(1-p)a}{(1-p)a+p(1-a)}$
\Else
\State $a\gets \frac{p a}{p a+(1-p)(1-a)}$
\EndIf
\If{$a\ge 1-\delta$}
\Return true
\Comment{$x<y$}
\EndIf
\If{$a\le \delta$}
\Return false
\Comment{$x>y$}
\EndIf
\EndWhile
\EndProcedure
\end{algorithmic}
\end{algorithm}

Because of the return condition, if prior distribution of $\mathbbm{1}\{x<y\}$ is $\Bern(\frac 12)$, then probability of error is $\le \delta$.
However, by symmetry, probability of error of the algorithm does not depend on the ground truth, so the probability of error is $\le \delta$ regardless of the ground truth.

In the following, let us analyze the number of queries used.

Let $a_t$ be the value of $a$ after the $t$-th query.
If the ground truth is $x<y$, then
\begin{align*}
\log \frac{a_{t+1}}{1-a_{t+1}} = \log \frac{a_t}{1-a_t} + \left\{\begin{array}{ll}
 + \log \frac{1-p}p & \text{w.p. $1-p$,}\\
 - \log \frac{1-p}p & \text{w.p. $p$.}
\end{array}\right.
\end{align*}

Thus,
\begin{align*}
\bE\left[\log \frac{a_{t+1}}{1-a_{t+1}}\right] = \log \frac{a_t}{1-a_t} + (1-2p) \log \frac{1-p}p,
\end{align*}
where the expectation is over the randomness of the $(t+1)$-th query.

Define two sequences $(X_t)_{t\ge 0}$, $(Y_t)_{t\ge 0}$ as $X_t = \log \frac{a_t}{1-a_t}$ and $Y_t = X_t - t(1-2p) \log \frac{1-p}p$.
Let $\tau$ be the stopping time $\tau := \min\{t: a_t \le \delta \text{ or } a_t \ge 1-\delta\}$.

By the above discussion, $(Y_t)_{t\ge 0}$ is a martingale.
Using Optional Stopping Theorem, we have
\begin{align*}
    \bE[Y_{\tau\land n}] = \bE[Y_0] = 0
\end{align*}
for every $n \ge 0$, 
so
\begin{align*}
    \bE[X_{\tau\land n}] = \bE[\tau \land n] (1-2p) \log\frac{1-p}p.
\end{align*}
By the bounded convergence theorem, $\bE[X_{\tau\land n}]$ goes to $\bE[X_\tau]$ as $n\to \infty$.
By the monotone convergence theorem, $\bE[\tau \land n]$ goes to $\bE[\tau]$ as $n\to \infty$.
Therefore,
\begin{align*}
    \bE[X_\tau] = \bE[\tau] (1-2p) \log\frac{1-p}p,
\end{align*}
which leads to 
\begin{align*}
    \bE[\tau] = \frac{\bE[X_\tau]}{(1-2p) \log \frac {1-p}p} \le \frac 1{1-2p} \left \lceil \frac{\log \frac{1-\delta}{\delta}}{\log \frac{1-p}p} \right\rceil
\end{align*}
where we used the fact that $X_t$ is always an integer multiple of $\log \frac{1-p}{p}$ and hence $X_\tau \le \left \lceil \frac{\log \frac{1-\delta}{\delta}}{\log \frac{1-p}p} \right\rceil \cdot \log \frac{1-p}{p}$.

Thus, the algorithm stops within $\frac 1{1-2p} \left \lceil \frac{\log \frac{1-\delta}{\delta}}{\log \frac{1-p}p} \right\rceil$ queries in expectation.

Similarly, if the ground truth is $x>y$, the algorithm stops within the same expected number of queries.
This finishes the proof.
\end{proof}

For simplicity, we use $f_p(\delta)$ to denote $\frac 1{1-2p} \left \lceil \frac{\log \frac{1-\delta}{\delta}}{\log \frac{1-p}p} \right\rceil$ for $0<p<\frac 12$ and $\delta > 0$, which, by Lemma~\ref{lem:repeated_query}, upper bounds the expected number of comparisons needed by Algorithm~\ref{algo:less_than} to compare two elements with error probability $\le \delta$. When clear from the context, we drop the subscript $p$.

\subsection{Pairwise Independent Hashing}
During our algorithm, we sometimes run into the following situation: we need to run a certain algorithm $n$ times, where each instance uses $m = n^{c}$ random bits for some $0<c<1$. We can only afford $o(n\log n)$ fresh random bits in total (for derandomization purpose), so the instances need to share randomness in some way. On the other hand, we want randomness between any two instances to be independent of each other, in order for concentration bounds to hold. Therefore we need the following standard result on pairwise independent hashing.
\begin{lemma} [{\cite{carter1977universal}}] \label{lem:pairwise-indep-hash}
There exists a pairwise independent hash family of size $2^k-1$ using $k$ fresh random bits.
\end{lemma}
By using $m$ fully independent copies of pairwise independent hash families, we can achieve pairwise independence between instances using $O(m\log n)$ fresh random bits.

\section{Noisy Sorting Algorithm}
In this section, we will present our algorithm for noisy sorting and prove Theorem~\ref{thm:sort-upperbound}. We will first present the following randomized version of our algorithm, and then convert it to a deterministic one in Section~\ref{sec:noisy-sort-derandomize}.

\begin{thm}
\label{thm:sort-upperbound-randomized}
There exists a  randomized algorithm for \NoisySorting{n} with error probability $o(1)$ which always uses at most $$(1+o(1))\left(\frac{1}{I(p)} +\frac{1}{(1-2p) \log \frac {1-p}p}\right) n\log n$$ noisy comparisons, and uses $O(n)$ random bits in expectation.
\end{thm}

We start with a subroutine that sorts a list with a small number of inversion pairs.

\begin{lemma}
\label{lem:sort_inversion}
Fix any parameter $\sigma \in (0, 1)$ which can depend on $n$. Given a list of $n$ elements with $t$ inversion pairs, there exists a deterministic algorithm using noisy comparisons that sorts these $n$ elements with error probability $\le (n-1+t)\sigma$ using $(n-1+t)f(\sigma) + (n-1+t)\sigma n^2 f(\sigma)$ noisy queries in expectation. The algorithm does not have to know $t$ in advance.
\end{lemma}
\begin{proof}
Consider Algorithm \ref{algo:sort_inversion}, which is essentially insertion sort.
\begin{algorithm}[h]
\caption{} \label{algo:sort_inversion}
\begin{algorithmic}[1]
\Procedure{SortInversion}{$n, (a_i)_{i\in [n]}, p,\sigma$}
\For{$i=1\to n$}
\For{$j=i\to 2$}
\If{\textsc{LessThan}($a_j,a_{j-1},\sigma$)}
\State Swap $a_j$ and $a_{j-1}$
\Else
\State \textbf{break}
\EndIf
\EndFor
\EndFor
\State \Return $(a_i)_{i\in [n]}$
\EndProcedure
\end{algorithmic}
\end{algorithm}

By union bound, with probability $\ge 1-(n-1+t)\sigma$, the first $(n-1+t)$ calls to \textsc{LessThan} all return correctly. Conditioned on this happening, the algorithm acts exactly like a normal insertion sort, which halts after the first $(n-1+t)$ calls to \textsc{LessThan}, which takes $(n-1+t)f(\sigma)$ noisy comparisons in expectation. 

With probability $\le (n-1+t)\sigma$, the algorithm does not run correctly, but the algorithm still halts in at most $\binom n2f(\sigma) \le n^2 f(\sigma)$ queries in expectation.
\end{proof}

Next, we use Lemma~\ref{lem:sort_inversion} to construct the following subroutine. It is ultimately used to sort some small sets of elements in the main algorithm.

\begin{lemma}
\label{lem:weak_sort}
Fix any parameter $\delta\in (0, 1)$ which can depend on $n$. There is a randomized algorithm for \NoisySorting{n}  with error probability at most $\delta$ using  $O(n\log n) +  n f(\delta/n) + n^2 \delta f(\delta/n)$ noisy comparisons in expectation, and using $O(n \log n)$ random bits always.
\end{lemma}
\begin{proof}
We first use Corollary~\ref{cor:simple_sort} with error probability $P_e \le n^{-2}$, then use Lemma~\ref{lem:sort_inversion} with error probability $\sigma=\delta/n$.

After the first step, the expected number of inversion pairs is
\begin{align*}
\bE[t] \le P_e n^2 \le 1,
\end{align*}
so the second step uses at most $n f(\delta/n) + n^2 \delta f(\delta/n)$ queries in expectation.

The overall error probability is at most $\bE[(n-1+t)\sigma] \le n \sigma = \delta$.
\end{proof}

The following lemma allows us to construct algorithms with guaranteed tail behavior for the number of queries. This will be helpful for concentration bounds.
\begin{lemma} \label{lem:variance_bound}
Suppose we have an algorithm $\cA$ (solving a certain task) that has error probability $\le \delta$ and number of queries $\tau$ with $\bE[\tau] \le m$, and uses $\le r$ random bits always.
Then we can construct an algorithm $\cB$ solving the same task satisfying the following properties:
\begin{itemize}
    \item $\cB$ has error probability $(1+o_m(1))\delta$;
    \item Let $\rho$ be the number of queries algorithm $\cB$ uses. Then 
    \begin{align*}
        \bE[\rho] &\le (1+o_m(1))m,\\
        \bE[\rho^2] & = O(m^3).
    \end{align*}
    \item Let $r'$ be the number of random bits algorithm $\cB$ uses. Then
    \begin{align*}
        \bE[r'] &\le (1+o_m(1)) r, \\
        \bE[(r')^2] &= O(r^2).
    \end{align*}
\end{itemize}
\end{lemma}
\begin{proof}
Let $k$ be a parameter to be chosen later.
Consider the following algorithm $\cB$:
\begin{enumerate}
    \item Run algorithm $\cA$ until it halts, or it is about to use the $(k+1)$-th query.
    \item If $\cA$ halts, then we return the return value of $\cA$; otherwise, we restart the whole algorithm.
\end{enumerate}

Let us compute the probability of a restart.
By Markov's inequality,
\begin{align*}
\bP[\text{restart}] = \bE[\tau > k] \le \frac mk.
\end{align*}

Algorithm $\cB$'s error probability $P_e(\cB)$ is
\begin{align*}
P_e(\cB) = \bP[\text{$\cA$ errs} \mid \tau \le k]
\le \frac{\bP[\text{$\cA$ errs}]}{\bP[\tau \le k]}  \le \frac{\delta}{1-\frac mk}.
\end{align*}

The expected number of queries $\rho$ used by Algorithm $B$ satisfies
\begin{align*}
\bE[\rho] &= \bE[\tau \land k] + \bP[\text{restart}] \bE[\rho] \\
& \le m + \frac mk \bE[\rho].
\end{align*}
Solving this, we get
\begin{align*}
\bE[\rho] \le \frac{m}{1-\frac mk}.
\end{align*}

The second moment satisfies
\begin{align*}
& \bE[\rho^2] \le \bE[(\tau\land k)^2] + \bP[\text{restart}] \bE[(k+\rho)^2] \\
& \le k^2 + \frac mk \cdot (\bE[\rho^2] + 2k \bE[\rho] + k^2) \\
& \le k^2 + mk + \frac{2m^2}{1-\frac mk} + \frac mk \cdot \bE[\rho^2].
\end{align*}
Solving this, we get
\begin{align*}
\bE[\rho^2] \le \left(1-\frac mk\right)^{-1} \left( k^2 + mk + \frac{2m^2}{1-\frac mk}\right).
\end{align*}

Similar to $\rho$, we have
\begin{align*}
    \bE[r'] &\le r + \bP[\text{restart}] \bE[r'] \\
    &\le r + \frac mk \cdot \bE[r'].
\end{align*}
Solving this we get
\begin{align*}
    \bE[r'] \le \frac {r}{1-\frac mk}.
\end{align*}
Also,
\begin{align*}
    \bE[(r')^2] &\le r^2 + \bP[\text{restart}] \bE[(r+r')^2] \\
    & \le r^2 + \frac mk (r^2 + 2r\bE[r'] + \bE[(r')^2]) \\
    & \le r^2 + \frac mk \cdot r^2 + \frac mk \cdot \frac{2r^2}{1-\frac mk} + \frac mk \cdot \bE[(r')^2].
\end{align*}
Solving this we get
\begin{align*}
    \bE[(r')^2] \le (1-\frac mk)^{-1} (r^2 + \frac mk \cdot r^2 + \frac mk \cdot \frac{2r^2}{1-\frac mk}).
\end{align*}

Choosing $k = m\log m$, we have $1/ (1-m/k) = 1+o_m(1)$, so 
we get
\begin{align*}
P_e(\cB) & \le (1+o_m(1))\delta,\\
\bE[\rho] & \le (1+o_m(1)) m,\\
\bE[\rho^2] & \le O(m^2\log^2 m) = O(m^3),\\
\bE[r'] & \le (1+o_m(1)) r,\\
\bE[(r')^2] &\le O(r^2).
\end{align*}
\end{proof}

Using Lemma~\ref{lem:variance_bound}, we are able to construct ``safe'' (i.e., with guaranteed tail behavior) versions of algorithms introduced before.
\begin{coro}[Safe Algorithms]
\label{coro:safe_algo}\
\begin{enumerate}
    \item Given $\delta\in (0, 1)$, there exists a randomized algorithm \textsc{SafeBinarySearch} for \NoisyBinarySearch{n}, with error probability $\le \delta$,  with the number of queries $\tau$ satisfying
    \begin{align*}
    \bE[\tau] = (1+o(1))\left(\frac{\log n + O(\log \frac 1\delta)}{I(p)}\right), \quad \Var[\tau] = O(\log^3 n + \log^3 \frac 1\delta).
    \end{align*}
    and with the number of used random bits $r$ satisfying
    \begin{align*}
    \bE[r] = O(\log n), \quad \Var[r] = O(\log^2 n).
    \end{align*}
    \item Given $\delta\in (0, 1)$, there exists a deterministic algorithm \textsc{SafeLessThan} which compares two elements with error probability $\le \delta$, with the number of queries $\tau$ satisfying
    \begin{align*}
    \bE[\tau] = (1+o_{1/\delta}(1)) \frac{\log \frac 1\delta}{(1-2p)\log \frac{1-p}p}, \quad
    \Var[\tau] = O(\log^3 \frac 1\delta).
    \end{align*}
    \item Given $\delta = O(\frac 1{n^3})$, there exists a randomized algorithm \textsc{SafeWeakSort} for \NoisySorting{n}, with error probability $\le \delta$,  with the number of queries $\tau$ satisfying 
    \begin{align*}
    \bE[\tau] = O(n\log n) + (1+o(1)) \frac{n \log \frac 1{\delta}}{(1-2p)\log \frac{1-p}p}, \quad
    \Var[\tau] = O(n^3 \log^3 \frac 1\delta).
    \end{align*}
    and with the number of used random bits $r$ satisfying
    \begin{align*}
    \bE[r] = O(n \log n), \quad
    \Var[r] = O(n^2\log^2 n).
    \end{align*}
\end{enumerate}
\end{coro}
\begin{proof}
Apply Lemma~\ref{lem:variance_bound} to Theorem~\ref{thm:BHbinary}, Lemma~\ref{lem:repeated_query}, and Lemma~\ref{lem:weak_sort} respectively.
\end{proof}

We introduce our last subroutine below, which is used to sort a subset of $\Theta(\frac{n}{\log n})$ elements in the main algorithm. 

\begin{lemma}\label{lem:safe_simple_sort}
There exists a randomized algorithm \textsc{SafeSimpleSort} for \NoisySorting{n} with error probability $o(1)$ which always uses $O(n\log n)$ queries and $O(n \log n)$ random bits always.
\end{lemma}
\begin{proof}
Consider the following algorithm $\cA$:
\begin{itemize}
\item 
Keep an array of the first $i$ elements for $i$ from $1$ to $n$. Each time a new element needs to be inserted, the algorithm uses \textsc{SafeBinarySearch} to find the correct position, with error probability $\frac 1{n\log n}$. %
\item Output the resulting array.
\end{itemize}

We have $(n-1)$ calls to \textsc{SafeBinarySearch}.
Let $E_i$ be the event that the $i$-th call to \textsc{SafeBinarySearch} returns the correct value, and let $E$ be the event that all $E_i$ happens. By union bound, probability of error $\bP[\neg E]$ is at most $(n-1)\cdot \frac 1{n\log n} = o(1)$.

Let $\tau_i$ be the number of queries used in the $i$-th call to \textsc{SafeBinarySearch}.
By Corollary~\ref{coro:safe_algo} Part 1, we have
\begin{align*}
\bE[\tau_i \mid E] = \bE\left[\tau_i \mid \bigwedge_{1 \le j \le i} E_j\right] \le \frac{\bE\left[\tau_i \mid \bigwedge_{1 \le j < i} E_j\right]}{\bP\left[E_i \mid \bigwedge_{1 \le j < i} E_j\right]}&= O(\log n),\\
\end{align*}
and
\begin{align*}
\Var[\tau_i \mid E] &= \Var\left[\tau_i \mid \bigwedge_{1 \le j \le i} E_j\right] \le \bE\left[\tau_i^2 \mid \bigwedge_{1 \le j \le i} E_j\right] \le \frac{\bE\left[\tau_i^2 \mid \bigwedge_{1 \le j < i} E_j\right]}{\bP\left[E_i \mid \bigwedge_{1 \le j < i} E_j\right]}\\
& \le O\left(\Var\left[\tau_i \mid \bigwedge_{1 \le j < i} E_j\right] + \bE\left[\tau_i \mid \bigwedge_{1 \le j < i} E_j\right]^2\right) = O(\log^3 n). 
\end{align*}
Conditioned on $E$, $\tau_i$'s are independent because they use disjoint source of randomness.
Thus, by Chebyshev's inequality,
\begin{align*}
\bP\left[\sum_{i\in [n-1]} \tau_i \ge \sum_{i\in [n-1]} \bE[\tau_i] + n^{2/3} \mid E \right] \le O\left(\frac{n\log^3 n}{\left(n^{2/3}\right)^2}\right) = o(1).
\end{align*}

Let us define $m = \sum_{i\in [n-1]} \bE[\tau_i \mid E] + n^{2/3} = O(n\log n)$. 

If the number of random bits used in the $i$-th call to \textsc{SafeBinarySearch} is $r_i$, then we can similarly show 

\begin{align*}
\bP\left[\sum_{i\in [n-1]} r_i \ge \sum_{i\in [n-1]} \bE[r_i] + n^{2/3} \mid E \right] = o(1).
\end{align*}
Let $R = \sum_{i\in [n-1]} \bE[r_i \mid E] + n^{2/3} = O(n\log n)$.

Consider the following algorithm (which is our \textsc{SafeSimpleSort}):
\begin{itemize}
    \item Run $\cA$ until it finishes, or it is about the make the $(m+1)$-th query, or it is about the use the $(R+1)$-th random bit. 
    \item If $\cA$ finishes, then return the output of $\cA$; otherwise, output an arbitrary permutation.
\end{itemize}
Then \textsc{SafeSimpleSort} always makes at most $m=O(n\log n)$ queries, uses at most $R = O(n\log n)$ random bits, and has error probability $\bP[\cA \text{ errs}] + o(1) = o(1)$.
\end{proof}

Finally, we give our main algorithm for noisy sorting. 
See Algorithm~\ref{algo:noisysort} for its description. 

We first analyze its error probability.

\begin{algorithm}[hbt!]
\caption{}
\label{algo:noisysort}
\begin{algorithmic}[1]
\Procedure{NoisySort}{$A=(x_1,\ldots,x_n), p$} 
\State $S \gets \{-\infty, \infty\}$.
\For{$a \in A$}
\State Add $a$ to $S$ with probability $\frac{1}{\log n}$. \Comment{Fully independent random bits.}
\label{line:sampleS}
\EndFor
\State Sort $S$ using \textsc{SafeSimpleSort}. \Comment{Fully independent random bits.}
\label{line:sortS}
\For{$a \in A \setminus S$}
\State Use \textsc{SafeBinarySearch} with $\delta = \frac{1}{\log n}$ to search the predecessor of $a$ in $S$. \Comment{We use $n^{2/3}$ pairwise independent hash families to feed the first $n^{2/3}$ random bits of each \textsc{SafeBinarySearch} call, and errs if some call needs more than $n^{2/3}$ random bits. }
\label{line:BH_for_a}
\State Denote the returned answer as $\hat{l}_a$. 
\EndFor
\State $X \gets \emptyset$.
\State $\mathcal{A} \gets S \setminus \{-\infty, \infty\}$.
\For{$l \in S \setminus \{\infty\}$}
\State $r \gets $ next element in $S$.
\State $B_l \gets \{a \in A \setminus S: \hat{l}_a = l\}$.
\If{$|B_l| > 6 \log^2 n$}
\State $X \gets X \cup B_l$
\label{line:addBtoX}
\Else
\State Sort $B_l$ using \textsc{SafeWeakSort} with error probability $\frac{1}{n \log n}$. \Comment{We use $n^{2/3}$ pairwise independent hash families to feed the first $n^{2/3}$ random bits of each \textsc{SafeWeakSort} call, and errs if some call needs more than $n^{2/3}$ random bits. }
\label{line:sortB}
\While{$|B_l| > 0$}
\State $x \gets $ first element in $B_l$.
\If{\textsc{SafeLessThan}($x$, $l$, $\frac{1}{n \log n}$)}
\label{line:lessthan1}
\State $B \gets B \setminus \{x\}$.
\State $X \gets X \cup \{x\}$.
\label{line:addxtoX1}
\Else
\State \textbf{break}
\EndIf
\EndWhile
\While{$|B_l| > 0$}
\State $x \gets $ last element in $B_l$.
\If{\textsc{SafeLessThan}($r$, $x$, $\frac{1}{n \log n}$)}
\label{line:lessthan2}
\State $B_l \gets B_l \setminus \{x\}$.
\State $X \gets X \cup \{x\}$.
\label{line:addxtoX2}
\Else
\State \textbf{break}
\EndIf
\EndWhile
\State Add $B_l$ to $\mathcal{A}$ between $l$ and $r$. 
\label{line:add_B_l}
\EndIf
\EndFor
\For{$x \in X$}
\label{line:for_x_in_X}
\State Insert $x$ to $\mathcal{A}$ using \textsc{SafeBinarySearch} with $\delta = \frac{1}{n \log n}$. \Comment{We use $n^{2/3}$ pairwise independent hash families to feed the first $n^{2/3}$ random bits of each \textsc{SafeBinarySearch} call, and errs if some call needs more than $n^{2/3}$ random bits. }
\label{line:insertX}
\EndFor
\State \Return $\mathcal{A}$. 
\EndProcedure
\end{algorithmic}
\end{algorithm}
\begin{lemma}
\label{lem:noisy_sort_error}
Algorithm~\ref{algo:noisysort} has error probability $o(1)$.
\end{lemma}
\begin{proof}
Let $E_0$ be the event that the algorithm does not err because of lack of random bits at Line~\ref{line:BH_for_a}, Line~\ref{line:sortB}, or Line~\ref{line:insertX}. 
Let $E_1$ be the event that \textsc{SafeSimpleSort} successfully sorts $S$ at Line~\ref{line:sortS}; $E_2$ be the event that all calls of \textsc{SafeWeakSort} at Line~\ref{line:sortB} are correct; $E_3$ be the event that all calls of \textsc{SafeLessThan} at Lines~\ref{line:lessthan1} and Lines~\ref{line:lessthan2} are correct; $E_4$ be the event that all calls of \textsc{SafeBinarySearch} at Line~\ref{line:insertX} are correct.

We show that if $E_0, E_1, E_2, E_3, E_4$ all happen, then Algorithm~\ref{algo:noisysort} is correct. First of all, under $E_1$, $S$ is correctly sorted at Line~\ref{line:sortS}. Secondly, under $E_2$ and $E_0$, for each $l$, $B_l$ is correctly sorted at Line~\ref{line:sortB}. Under $E_3$ and $E_0$, at Line~\ref{line:add_B_l}, all elements remaining in $B_l$ are indeed greater than $l$ and smaller than $r$, so adding $B_l$ to $\mathcal{A}$ between $l$ and $r$ keeps all elements in $\mathcal{A}$ sorted. Therefore, before the for loop at Line~\ref{line:for_x_in_X}, all elements in $\mathcal{A}$  are correctly sorted, and clearly these elements are exactly $A \setminus X$. Finally, under $E_4$ and $E_0$, the insertions made at Line~\ref{line:insertX} are all correct, so the final result is correct. 

For each call of \textsc{SafeBinarySearch} at Line~\ref{line:BH_for_a}, the probability that it requires more than $n^{2/3}$ random bits is $O(\frac{\polylog n}{n^{4/3}})$ by Corollary~\ref{coro:safe_algo} Part 1 and Chebyshev's inequality. By union bound, with probability $1-o(1)$, none of the calls of \textsc{SafeBinarySearch} at Line~\ref{line:BH_for_a} causes the algorithm to err. We can similarly bound the probability the algorithm errs at Line~\ref{line:sortB} or Line~\ref{line:insertX}. 
Overall, $\Pr[\neg E_0] \le o(1)$. 
Clearly, $\Pr[\neg E_1] \le o(1)$, $\Pr[\neg E_2] \le \frac{1}{n \log n} \cdot (|S|-1)\le \frac{1}{\log n}$, $\Pr[\neg E_3] \le \frac{1}{n \log n} \cdot (2n) = \frac{2}{\log n}$ and $\Pr[\neg E_4] \le \frac{1}{n \log n} \cdot n = \frac{1}{\log n}$. Thus, by union bound, the overall success probability of Algorithm~\ref{algo:noisysort} is at least $1-\frac{4}{\log n} -o(1) = 1 - o(1)$. 
\end{proof}

Let a bucket be the set of elements that are between two adjacent elements in the sorted order of $S$. Then we have the following simple lemma.
\begin{lemma} \label{lem:bucket_size}
With probability $\ge 1-\frac{1}{n}$, all buckets have size at most $3 \log^2 n$.
\end{lemma}
\begin{proof}
For any continuous segment of length $L$ in the sorted order of $A$, the probability that none of these element is in $S$ is
$(1-1/\log n)^L \le \exp(-L/\log n)$.
By union bound, the probability that there exists $L$ continuous elements not in $S$ is at most $n\exp(-L/\log n)$.
Taking $L = 2\log_e n \cdot \log n \le 3\log^2 n$, we see that with probability $\ge 1-\frac 1n$, all continuous segments of length $L$ have at least one element in $S$, which implies all  buckets have size $\le 3\log^2 n$.
\end{proof}

\begin{lemma}
\label{lem:noisy_sort_time}
With probability $1-o(1)$, Algorithm~\ref{algo:noisysort} uses at most $(1+o(1)) \left(\frac{1}{I(p)} + \frac{1}{(1-2p) \log \frac {1-p}p}\right)n\log n$ queries.
\end{lemma}
\begin{proof}
Because erring only makes the algorithm exit earlier, in the analysis we ignore the effect of erring caused by insufficient random bits.
We define the following events:
\begin{itemize}
    \item Let $E_1$ be the event that $|S| \le \frac n{\log n} + n^{2/3} = (1+o(1)) \frac n{\log n}$. By Chernoff bound, $\bP[E_1] = 1-o(1)$.
    \item Let $E_2$ be the event that all buckets have size at most $3\log^2 n$. By Lemma~\ref{lem:bucket_size}, $\bP[E_2] = 1-o(1)$.
    \item Let $E_3$ be the event that at most $\frac n{\log n} + n^{2/3} = (1+o(1)) \frac n{\log n}$ elements $a$ have the wrong predecessor $\hat{l}_a$ at Line~\ref{line:BH_for_a}. By Chebshev's inequality, $\bP[E_3] = 1-o(1)$.
    \item Let $E_4$ be the event that all calls to \textsc{SafeLessThan} on Line~\ref{line:lessthan1}, \ref{line:lessthan2} return the correct values. Because there are at most $2n$ calls, by union bound, $\bP[E_4] = 1-o(1)$.
    \item Let $E_5$ be the event that $|X| \le O(\frac n{\log n})$.
    If $E_2$, $E_3$, and $E_4$ all happen, then $E_5$ happens, for the following reasons. First, at Line~\ref{line:addBtoX}, $B_l$ has size greater than $6 \log^2 n$, while the sizes of all buckets are at most $3 \log^2 n$ conditioned on $E_2$. Thus, at least half of the elements in $B_l$ have the wrong predecessor $\hat{l}_a$ at Line~\ref{line:BH_for_a}. Thus, the amount of elements added to $X$ at Line~\ref{line:addBtoX} is at most twice the number of elements $a$ with the wrong predecessor $\hat{l}_a$, which is bounded by $O(\frac n{\log n})$ conditioned on $E_3$. Also, if an element $x$ is added to $X$ at Line~\ref{line:addxtoX1} or \ref{line:addxtoX2} and $E_4$ happens, $x$ must also have the wrong predecessor. Thus, overall, $|X| \le O(\frac n{\log n})$ if $E_2$, $E_3$, and $E_4$ all happen. 
\end{itemize}

We make queries in the following lines: Line~\ref{line:sortS}, \ref{line:BH_for_a}, \ref{line:sortB}, \ref{line:lessthan1}, \ref{line:lessthan2}, \ref{line:insertX}. Let us consider them separately.

    \paragraph{Line~\ref{line:sortS}.} By Lemma~\ref{lem:safe_simple_sort}, conditioned on $E_1$, with probability $1-o(1)$, Line~\ref{line:sortS} uses $O(n)$ queries. 
    \paragraph{Line~\ref{line:BH_for_a}.} We have at most $n$ calls to \textsc{SafeBinarySearch} whose number of queries are independent.
    By Corollary~\ref{coro:safe_algo} Part 1, $$\bE[\text{number of queries}] = (1+o(1)) \frac{n\log n}{I(p)}.$$
    By Corollary~\ref{coro:safe_algo} Part 1 and Chebyshev's inequality, $$\bP\left[\text{number of queries} > \bE[\text{number of queries}\right]+ n^{2/3}] = o(1).$$
    \paragraph{Line~\ref{line:sortB}.} We have a few calls to \textsc{SafeWeakSort} where every input length is at most $6\log^2 n$, and total input length is at most $n$.
    By Corollary~\ref{coro:safe_algo} Part 3, $$\bE[\text{number of queries}] = (1+o(1)) \frac{n\log n}{(1-2p)\log \frac{1-p}p}.$$
    By Corollary~\ref{coro:safe_algo} Part 3 and Chebyshev's inequality, $$\bP[\text{number of queries} > \bE[\text{number of queries}]+ n^{2/3}] = o(1).$$
    \paragraph{Line~\ref{line:lessthan1}, \ref{line:lessthan2}.}
    Conditioned on $E_1$, $E_3$, $E_4$, the number of calls to \textsc{SafeLessThan} is 
    \begin{align*}
    O(|S|) + O(\text{number of elements in the wrong bucket}) = O(\frac n{\log n}).
    \end{align*}
    By Corollary~\ref{coro:safe_algo} Part 2, $$\bE[\text{number of queries}] = O(n).$$
    By Corollary~\ref{coro:safe_algo} Part 2 and Chebyshev's inequality, $$\bP[\text{number of queries} > \bE[\text{number of queries}]+ n^{2/3}] = o(1).$$
    \paragraph{Line~\ref{line:insertX}.} Conditioned on $E_5$, the number of calls to \textsc{SafeBinarySearch} is $O(\frac n{\log n})$.
    By Corollary~\ref{coro:safe_algo} Part 1, $$\bE[\text{number of queries}] = O(n)$$
    By Corollary~\ref{coro:safe_algo} Part 1 and Chebyshev's inequality, $$\bP[\text{number of queries} > \bE[\text{number of queries}]+ n^{2/3}] = o(1).$$

Summarizing the above, excluding events with $o(1)$ probability in total, the total number of queries made is
\begin{align*}
& O(n)+ (1+o(1)) \frac{n\log n}{I(p)} + (1+o(1)) \frac{n\log n}{(1-2p)\log \frac{1-p}p} + O(n) + O(n) + O(n^{2/3}) \\
& = (1+o(1))\left(\frac{n\log n}{I(p)} +\frac{n\log n}{(1-2p) \log \frac {1-p}p}\right).
\end{align*}
\end{proof}

Finally, we analyze the expected number of random bits used by the algorithm.

\begin{lemma} \label{lem:noisy_sort_random_bit}
Algorithm~\ref{algo:noisysort} uses $O(n)$ random bits in expectation. 
\end{lemma}
\begin{proof}
We consider all the places Algorithm~\ref{algo:noisysort} uses randomness. 

\paragraph{Line~\ref{line:sampleS}.}
Any biased random bits can be simulated by $O(1)$ random bits in expectation \cite{KY76}, so Line~\ref{line:sampleS} uses $O(n)$ random bits in expectation. 
\paragraph{Line~\ref{line:sortS}.} By Chernoff bound, $|S| \le 100 n / \log n$ with probability at least $1-n^{-2}$. In this case, \textsc{SafeSimpleSort} uses $O(|S| \log |S|) = O(n)$ random bits. If $S >100 n / \log n$, \textsc{SafeSimpleSort} uses $O(n \log n)$ random bits, which contribute at most  $O(n \log n \cdot n^{-2})$ to the overall expectation. Thus, Line~\ref{line:sortS} uses $O(n)$ random bits in expectation. 
\paragraph{Line~\ref{line:BH_for_a}, Line~\ref{line:sortB}, Line~\ref{line:insertX}. }
For each of the first $n^{2/3}$ outputted random bits, we use Lemma~\ref{lem:pairwise-indep-hash} to construct a pairwise independent hash family of size $n$ using $O(\log n)$ fresh random bits.
The total number of fresh random bits needed is $O(n^{2/3} \log n)$.
\end{proof}

Now we are ready to prove Theorem~\ref{thm:sort-upperbound-randomized}.
\begin{proof}[Proof of Theorem~\ref{thm:sort-upperbound-randomized}]
Consider the following algorithm, which we call \textsc{SafeNoisySort}:
\begin{itemize}
    \item Run \textsc{NoisySort}. Stop immediately if we are about to use $m$ queries, for some $m$ to be chosen later.
    \item If we successfully completed \textsc{NoisySort}, then output the sorted array; otherwise, output any permutation.
\end{itemize}

By Lemma~\ref{lem:noisy_sort_time}, we can choose $m=(1+o(1)) \left(\frac{1}{I(p)} + \frac{1}{(1-2p) \log \frac {1-p}p}\right)n\log n$ so that with probability $1-o(1)$, the \textsc{NoisySort} call completes successfully. By Lemma~\ref{lem:noisy_sort_error}, with probability $1-o(1)$, the \textsc{NoisySort} call is correct. 
By union bound, \textsc{SafeNoisySort} has error probability $o(1)$ and it always takes at most $m$ queries.
By Lemma~\ref{lem:noisy_sort_random_bit}, number of random bits used is $O(n)$ in expectation.
\end{proof}

\subsection{Trading Queries for Random Bits} \label{sec:noisy-sort-derandomize}
In this section, we show how to use the randomized algorithm to construct a deterministic one.

\SortUpperBound*
\begin{proof}
In our model, it is possible to generate unbiased random bits using noisy comparisons. Take two arbitrary elements $x, y$, and call $\textsc{NoisyComparison}(x,y)$ twice. If the two returned values are different, we use the result of the first returned value as our random bit; otherwise, we repeat. Clearly, this procedure generates an unbiased random bit. The probability that this procedure successfully return a random bit after each two calls of \textsc{NoisyComparison} is $2p(1-p)$, so the expected number of noisy comparisons needed to generate each random bit is $O(\frac{1}{2p(1-p)}) = O(1)$. Note that this procedure is deterministic.  

Algorithm~\ref{algo:noisysort} uses $O(n)$ random bits in expectation, so we can use the above procedure to generate random bits, and the expected number of queries needed is $O(n)$. We can halt the algorithm if the number of queries used this way exceeds $O(n \log \log n)$, which only incurs an additional $o(1)$ error probability.  
\end{proof}

\section{Noisy Sorting Lower Bound}
In this section, we show the lower bound of noisy sorting.
\SortLowerBound*

The following problem serves as an intermediate step towards the lower bound. 

\begin{defn}[Group Sorting Problem]
For some integers $k\mid n$, we define problem \GroupSorting{n}{k} as follows:
Given a list $L$ of $n$ elements, divided into $n/k$ groups $A_1,\ldots, A_{n/k}$ of $k$ elements each, satisfying the property that for all $1\le i<j\le n/k$, $x\in A_i, y\in A_j$, we have $x<y$.
An algorithm needs to output an ordered list of sets $A_1,\ldots, A_{n/k}$ by asking the following type of queries $\textsc{GroupQuery}(x,y)$:
\begin{itemize}
    \item If the two elements are in the same group, then the query result is $*$.
    \item Otherwise, suppose $x\in A_i$, $y\in A_j$ with $i\ne j$. Then the query result is \NoisyComparison{i}{j}.
\end{itemize}
\end{defn}

The following lemma relates \textsc{GroupSorting} and \textsc{NoisySorting}.
\begin{lemma} \label{lemma:noisysort-to-groupsort}
Fix any algorithm $\cA$ for \NoisySorting{n} with error probability $o(1)$.
Let $k\mid n$.
Let $L$ be the input list of \NoisySorting{n}. Suppose we partition $L$ into $n/k$ groups $A_1,\ldots,A_{n/k}$ where for all $1\le i<j \le n/k$, for all $x\in A_i$, $y\in A_j$, we have $x<y$. Let $U^{\ne}$ denote the number of queries $\cA$ makes to elements in different groups.

Then there exists an algorithm for \GroupSorting{n}{k} with error probability $o(1)$ which makes at most $\bE[U^{\ne}] + (n-n/k)$ queries in expectation.
\end{lemma}
\begin{proof}
Given $\cA$, we design an algorithm $\cA'$ for \GroupSorting{n}{k} as follows.

Given input $L$, Algorithm $\cA'$ picks an arbitrary strict total order $<_T$ on all the elements. 
$\cA'$ also maintains which elements are known to be in the same group, implied by queries that return $*$. 
Then, it simulates Algorithm $\cA$ on the same input $L$ as follows:
\begin{itemize}
\item When $\cA$ attempts to make a comparison between $x$ and $y$:
\begin{itemize}
    \item If $x, y$ are not known to be in the same group, call $\textsc{GroupQuery}(x,y)$. If the query returns $*$, which means $x$ and $y$ are in the same group, return $\BSC_p(\mathbbm{1} \{x <_T y\})$ to $\cA$. Otherwise, return the result of $\textsc{GroupQuery}(x,y)$ to $\cA$.
    \item If $x, y$ are known to be in the same group, return $\BSC_p(\mathbbm{1} \{x <_T y\})$ to $\cA$.
\end{itemize}

\item When $\cA$ outputs a sequence $x_1,\ldots, x_n$:
Let $A_i = \{x_{(i-1)k+1},\ldots, x_{ik}\}$ for $i\in [n/k]$ and output $A_1,\ldots, A_{n/k}$.

\end{itemize}
Let us analyze the error probability and number of queries made by algorithm $\cA'$.

\paragraph{Error probability.} Algorithm $\cA'$ simulates a \NoisySorting{n} instance on $L$ with the following total order:
\begin{itemize}
    \item If $x,y$ are in the same group, then $x<y$ iff $x<_T y$;
    \item If $x\in A_i$, $y\in A_j$ are in different groups, then $x<y$ iff $i<j$.
\end{itemize}
Therefore, with error probability $o(1)$, the sequence $x_1, \ldots, x_n$ that $\cA$ outputs is sorted with respect to the above total order, which gives the valid groups.
\paragraph{Number of queries.}
Algorithm $\cA'$ makes one query when $\cA$ makes a query for elements in different groups.
Algorithm $\cA'$ makes queries between elements in the same group only when they are not known to be in the same group. Imagine an initially empty graph on vertex set $A$ where for each $\textsc{GroupQuery}(x,y)$ that returns $*$, we add an edge between $x$ and $y$. 
For every such query, the number of connected components in this graph decreases by $1$, and at the end the graph has at least $n/k$ connected components. 
Thus, the total number of queries made by $\cA'$ to elements in the same group is at most $n-n/k$.

Overall, the total number of queries made by Algorithm $\cA'$ is at most $\bE[U^{\ne}] + (n-n/k)$ in expectation.
\end{proof}

Next we prove a lower bound for the Group Sorting Problem. By Lemma~\ref{lemma:noisysort-to-groupsort}, this implies a lower bound of the number of queries made between elements in different groups of any Noisy Sorting algorithm. 

\begin{lemma} \label{lemma:groupsort-lowerbound}
Let $k\mid n$ with $k = n^{o(1)}$.
Any (deterministic or randomized) algorithm for \GroupSorting{n}{k} with error probability $o(1)$ makes at least $(1-o(1)) \frac{n\log n}{I(p)}$ queries in expectation, even if the input is a uniformly random permutation. 
\end{lemma}

\begin{proof}
Fix an algorithm $\cA$ for \GroupSorting{n}{k}.
WLOG assume that %
$\cA$ makes queries between elements in the same group only when they are not known to be in the same group.
Therefore, the total number of queries between elements in the same group is at most $n-n/k$.

Let $X = (A_1,\ldots,A_{n/k})$ be the true partition.
Let $m$ be the (random) number of queries made.
Let $Q^m$ be the queries made.
Let $Y^m$ be the returned values of the queries.
Let $\hat X$ be the most probable $X$ given all query results. Note that $\hat X$ is a function of $(Q^m, Y^m)$.

When $k=n^{o(1)}$, we have $H(X) = \log n! - \frac nk \log k! = (1-o(1)) n\log n$.

By Fano's inequality, $H(X\mid \hat X) \le 1 + o(n\log n)$,
so 
\begin{align}
\label{eq:groupsorting-fano}
I(X; Q^m, Y^m) \ge I(X; \hat X) = H(X) - H(X \mid  \hat X) \ge (1-o(1)) n\log n.
\end{align}

On the other hand,
\begin{align}
I(X; Q^m, Y^m) &= \sum_{i\ge 1} \bP[m \ge i] I(X; Q_i, Y_i \mid  Q^{i-1}, Y^{i-1}, m \ge i) \tag{chain rule of mutual information}\\
&= \sum_{i\ge 1} \bP[m \ge i] I(X;Y_i \mid  Q_i, Q^{i-1}, Y^{i-1}, m \ge i) \tag{$Q_i$ and $X$ are independent conditioned on $(Q^{i-1},Y^{i-1})$}\\
&\le \sum_{i\ge 1} \bP[m\ge i] I(X, Q_i, Q^{i-1}, Y^{i-1}; Y_i \mid  m \ge i) \tag{chain rule}\\
&= \sum_{i\ge 1} \bP[m\ge i] I(X, Q_i; Y_i \mid  m \ge i) \tag{$(Q^{i-1}, Y^{i-1})$ and $Y_i$ are independent conditioned on $(X, Q_i)$}\\
&= \sum_{i\ge 1} \bP[m\ge i] \left(H(Y_i \mid  m\ge i) - H(Y_i \mid  X, Q_i, m\ge i)\right)
\label{eq:groupsorting-eq2}
\end{align}

Let $q_i = \bP[Y_i=* \mid  m\ge i]$.
Because the algorithm makes at most $n-n/k$ queries between elements in the same group, we have
\begin{align*}
\sum_{i\ge 1} \bP[m\ge i] q_i \le n-n/k.
\end{align*}

We have
\begin{align*}
H(Y_i \mid  m\ge i) \le h\left(q_i, \frac{1-q_i}2, \frac{1-q_i}2\right)
\end{align*}
and
\begin{align*}
H(Y_i \mid  X, Q_i, m\ge i) = 0 \cdot q_i + h(p) (1-q_i).
\end{align*}

Note that $H(Y_i \mid  m\ge i) \le \log 3$ trivially, so by Equation~(\ref{eq:groupsorting-eq2}), $I(X; Q^m, Y^m) \le \log 3 \cdot \sum_{i \ge 1} \bP[m \ge i] = \log 3 \cdot \bE[m]$. Combining with Equation~(\ref{eq:groupsorting-fano}), we have $\bE[m] = \Omega(n\log n)$.

Note that $g(x) := h(x,\frac{1-x}2,\frac{1-x}2)$ is concave in $x$, so
\begin{align*}
\sum_{i\ge 1} \bP[m\ge i] H(Y_i \mid  m\ge i) &\le \left(\sum_{i\ge 1} \bP[m\ge i]\right) g\left(\frac{\sum_{i\ge 1} \bP[m\ge i] q_i}{\sum_{i\ge 1} \bP[m\ge i]}\right) \tag{Jensen's inequality}\\
& \le \bE[m] \cdot g\left(\min\left\{\frac 13, \frac{n-n/k}{\bE[m]}\right\}\right)  \tag{$g$ is increasing on $[0, \frac 13]$ and is maximized at $\frac 13$}\\
& = (1+o(1)) \bE[m].\tag{$n = o(\bE[m])$}
\end{align*}

On the other hand,
\begin{align*}
\sum_{i\ge 1} \bP[m\ge i] H(Y_i \mid  X, Q_i, m\ge i) &= \sum_{i\ge 1} \bP[m\ge i] h(p) (1-q_i) \\
&= h(p) \left(\sum_{i\ge 1} \bP[m\ge i] - \sum_{i\ge 1} \bP[m\ge i] q_i\right) \\
&\ge h(p) (\bE[m] - n) \\
&= (1-o(1)) h(p) \bE[m].
\end{align*}

Plugging the above two inequalities in Equation~(\ref{eq:groupsorting-eq2}), we get
\begin{align*}
I(X; Q^m, Y^m) &\le (1+o(1)) \bE[m] - (1-o(1)) h(p) \bE[m] \\
&\le (1+o(1)) I(p) \bE[m].
\end{align*}

Combining with Equation~(\ref{eq:groupsorting-fano}), we get $\bE[m] \ge (1-o(1))\frac{n\log n}{I(p)}$.
\end{proof}

Before we give lower bound for the number of queries made to elements in the same group, we first show
the following technical lemma that will be useful later.
\begin{lemma} \label{lemma:noisysort-samegroup-lowerbound-technical}
Let $(X_n)_{n\ge 0}$ be a Markov process defined as
\begin{align*}
    X_0=0, \quad X_{n+1} = X_n + \left\{ \begin{array}{ll}
    + 1 & \text{w.p.}~1-p, \\
    - 1 & \text{w.p.}~p.
    \end{array} \right.
\end{align*}
Let $\tau$ be a random variable supported on $\bZ_{\ge 0}$.
Suppose there exists $\delta>0$ such that
\begin{align*}
\bE[\tau] \le (1-\delta)m/(1-2p).
\end{align*}
Then
\begin{align*}
\bP[X_\tau \le m] \ge \delta/2 - o_m(1).
\end{align*}
\end{lemma}
\begin{proof}
    Let $E_1$ be the event that $\tau \le (1-\delta/2) m/(1-2p)$.
    Let $E_2$ be the event that $X_t \le m$ for all $t \le (1-\delta/2) m/(1-2p)$.
    By Markov's inequality, $\bP[E_1] \ge \delta/2$.
    By concentration inequalities, $\bP[E_2] \ge 1-o_m(1)$.
    Using union bound, we have
    \begin{align*}
    \bP[X_\tau \le m] \ge  \bP[E_1 \land E_2] \ge \delta/2 - o_m(1).
    \end{align*}
\end{proof}

\begin{lemma} \label{lemma:noisysort-samegroup-lowerbound}
Fix any (deterministic or randomized) algorithm for \NoisySorting{n} with error probability $o(1)$.
Let $k\mid n$.
Let $L$ be the input list of \NoisySorting{n}. Suppose we partition $L$ into $n/k$ groups $A_1,\ldots,A_{n/k}$ where for all $1\le i<j \le n/k$, for all $x\in A_i$, $y\in A_j$, we have $x<y$. Let $U^{=}$ denote the number of queries made to elements in the same group.

Then $\bE[U^{=}] \ge (1-o(1)) \frac{(n-n/k) \log(n-n/k)}{(1-2p) \log \frac {1-p}p}$, even if the input is a uniformly random permutation. 
\end{lemma}

\begin{proof}
Suppose the true order is $\sigma=(x_1,\ldots, x_n)$.
Let $S = \{i\in [n]: (n/k) \nmid i\}$. Note that $|S| = n-n/k$.
For $i\in S$, define $W_i$ to be the number of queries returning $x_i < x_{i+1}$, minus the number of queries returning $x_i > x_{i+1}$. Note that $W_i$ is a random variable depending on the transcript.
Define $W = \sum_{i\in S} W_i$.
Then $\bE[W] = (1-2p)\bE[U^{=}]$.

Let us consider the posterior probabilities $q$ of the permutations given a fixed transcript.
That is, define $$q_\tau := \bP[\sigma=\tau \mid  Q^m, Y^m].$$
By Bayes' rule, $q_\tau$ equals $\frac{\bP[Q^m, Y^m \mid \sigma = \tau]\bP[\sigma = \tau]}{\bP[Q^m, Y^m]}$. By symmetry, $\bP[\sigma = \tau]$ are equal for different $\sigma$ and $\bP[Q^m, Y^m]$ is constant as we fixed the transcript, so $q_\tau$ is proportional to $\bP[Q^m, Y^m \mid \sigma = \tau]$. If $\tau = (y_1, \ldots, y_n)$, then $\bP[Q^m, Y^m \mid \sigma = \tau]$ equals
\begin{equation}
\label{eq:post_proportion}
    \prod_{1 \le i < j \le n} (1-p)^{(\text{\#queries returning } y_i < y_{j})} \cdot \prod_{1 \le i < j \le n} p^{(\text{\#queries returning } y_i > y_{j})}.
\end{equation}
For $i\in S$, consider the permutation $\tau_i := (x_1,\ldots,x_{i-1}, x_{i+1},x_i,x_{i+2},\ldots,x_n)$.
Then by comparing $q_{\tau_i}$ and $q_\sigma$ using (\ref{eq:post_proportion}), we have 
\begin{align*}
\frac{q_\sigma}{q_{\tau_i}} &= \frac{(1-p)^{(\text{\#queries returning } x_i < x_{i+1})}p^{(\text{\#queries returning } x_i > x_{i+1})}}{(1-p)^{(\text{\#queries returning } x_i > x_{i+1})}p^{(\text{\#queries returning } x_i < x_{i+1})}}\\
&= \left(\frac{1-p}{p}\right)^{(\text{\#queries returning } x_i < x_{i+1})} \left(\frac{1-p}{p}\right)^{-(\text{\#queries returning } x_i > x_{i+1})}\\
&=  \left(\frac{1-p}{p}\right)^{W_i}.
\end{align*}
Thus, 
\begin{align*}
q_{\tau_i} = q_{\sigma} \cdot \left(\frac{1-p}p\right)^{-W_i}.
\end{align*}
Summing over $S$, we have
\begin{align*}
\sum_{i\in S} q_{\tau_i} &= q_{\sigma} \sum_{i\in S} \left(\frac{1-p}p\right)^{-W_i} \\
& \ge q_{\sigma} |S| \left(\frac{1-p}p\right)^{-W/|S|}.
\end{align*}
Because the error probability is $o(1)$, for any $\epsilon>0$, with probability $1-o(1)$, $q_{\sigma} \ge 1-\epsilon$.
So
\begin{align*}
|S| \left(\frac{1-p}p\right)^{-W/|S|} \le \epsilon/(1-\epsilon)
\end{align*}
with probability $1-o(1)$.

On the other hand, for any $\delta > 0$, if
\begin{align*}
\bE[U^{=}] \le (1-\delta)\frac{|S| \log |S|}{(1-2p) \log\frac{1-p}p},
\end{align*}
then by regarding $U^=$ as $\tau$ and $W$ as $X_\tau$ in Lemma~\ref{lemma:noisysort-samegroup-lowerbound-technical}, with probability $\Omega(\delta)-o_{|S|}(1)$, we have
\begin{align*}
W \le \frac{|S| \log |S|}{\log\frac{1-p}p}
\end{align*}
and then
\begin{align*}
|S| \left(\frac{1-p}p\right)^{-W/|S|} \ge 1,
\end{align*}
which is a contradiction.

Thus, 
\begin{align*}
\bE[U^{=}] \ge (1-o(1))\frac{|S| \log |S|}{(1-2p) \log\frac{1-p}p},
\end{align*}
as desired.
\end{proof}

Now we are ready to prove Theorem~\ref{thm:sort-lowerbound}.
\begin{proof}[Proof of Theorem~\ref{thm:sort-lowerbound}]
Fix an algorithm for \NoisySorting{n} with error probability $o(1)$.

Take $k =\log n$. Let $n' =  k \lfloor n / k\rfloor$. Clearly, an algorithm for \NoisySorting{n} can be used to solve $\mathsf{NoisySorting}(n')$ with less or equal number of queries in expectation. Let $U^{\ne}$, $U^{=}$ be as defined in Lemma~\ref{lemma:noisysort-to-groupsort}, Lemma~\ref{lemma:noisysort-samegroup-lowerbound} for $n'$ and $k$. 

By Lemma~\ref{lemma:noisysort-to-groupsort} and Lemma~\ref{lemma:groupsort-lowerbound}, we have $$\bE[U^{\ne}] \ge (1-o(1)) \frac{n'\log n'}{I(p)} - (n'-n'/k) = (1-o(1)) \frac{n'\log n'}{I(p)}.$$
By Lemma~\ref{lemma:noisysort-samegroup-lowerbound}, we have
$$\bE[U^{=}] \ge (1-o(1)) \frac{(n'-n'/k)\log(n'-n'/k)}{(1-2p)\log \frac{1-p}p} = (1-o(1)) \frac{n'\log n'}{(1-2p)\log \frac{1-p}p}.$$

Therefore expected number of queries made by the algorithm is at least
\begin{align*}
\bE[U^{\ne}] + \bE[U^{=}] & \ge (1-o(1)) \left(\frac 1{I(p)} + \frac 1{(1-2p) \log \frac{1-p}p}\right) n'\log n'\\
& = (1-o(1)) \left(\frac 1{I(p)} + \frac 1{(1-2p) \log \frac{1-p}p}\right) n\log n.
\end{align*}
\end{proof}

\section{Noisy Binary Search}

In this section, we prove our lower and upper bounds for Noisy Binary Search. 

\subsection{Lower Bound}
In this section, we prove our lower bound for Noisy Binary Search using techniques similar to our lower bound for Noisy Sorting.
\SearchLowerbound*

We define the following intermediate problem:
\begin{defn}[\OracleNoisyBinarySearch{n}]
Given a sorted list $L$ of $n$ elements and an element $x$ that is unequal to any element in the list, a solver needs to output the predecessor of $x$ in $L$ by asking the following type of query $\textsc{OracleQuery}(x,y)$ for $y \in L$:
\begin{itemize}
    \item If $y$ is the predecessor of $x$, output symbol $P$;
    \item If $y$ is the successor of $x$, output symbol $S$;
    \item Otherwise, output \NoisyComparison{x}{y}. 
\end{itemize}
\end{defn}

\begin{lemma} \label{lemma:binary-search-to-oracle-binary-search}
Fix an algorithm $\cA$ for \NoisyBinarySearch{n} with error probability $\delta$. Let $L$ be the input list and $x$ be the input element of \NoisyBinarySearch{n}. 
Let $U^{\not \approx}$ denote the number of noisy comparisons $\cA$ makes  between $x$ and elements that are not the predecessor or successor of $x$. 

Then there exists an algorithm for \OracleNoisyBinarySearch{n} with error probability at most $\delta$ which makes at most $\bE[U^{\not \approx}] + 1$ queries in expectation.
\end{lemma}
\begin{proof}
The algorithm for \OracleNoisyBinarySearch{n} simulates $\cA$ as follows:
\begin{itemize}
    \item Every time $\cA$ attempts to make a \NoisyComparison{x}{y} between some $x$ and some element $y$ in $L$, the algorithm makes $\textsc{OracleQuery}(x,y)$ instead. If $\textsc{OracleQuery}(x,y)$ returns $P$ or $S$, we have found the predecessor of $x$, so we can return the predecessor and finish the execution; otherwise, we pass the result of $\textsc{OracleQuery}(x,y)$ (which will be \NoisyComparison{x}{y}) to $\cA$. 
    \item If $\cA$ returns a predecessor or $x$, then we also return the same predecessor, finishing the execution.
\end{itemize}
The error probability bound and number of queries easily follow. 
\end{proof}

\begin{lemma} \label{lemma:oracle-binary-search-lowerbound}
Any (deterministic or randomized) algorithm for \OracleNoisyBinarySearch{n} with error probability $\le \delta$ makes at least $(1-\delta-o(1)) \frac{\log n}{I(p)}$ queries in expectation, even if the position of the element to search for is uniformly random. 
\end{lemma}
\begin{proof}
Fix an algorithm $\cA$ for \OracleNoisyBinarySearch{n}. WLOG, we can assume that $\cA$ exits as soon as it sees a query that returns $P$ or $S$, i.e., it sees at most one such query result. Also, we assume $\delta \le 1-\Omega(1)$, as otherwise the lower bound is trivially $0$. 

Let $X \in \{0, \ldots, n\}$ be the index of the true predecessor of $x$. Let $m$ be the (random) number of queries made.
Let $Q^m$ be the queries made.
Let $Y^m$ be the returned values of the queries.
Let $\hat X$ be the most probable $X$ given all query results. Note that $\hat X$ is a function of $(Q^m, Y^m)$.

Clearly, $H(X) = \log(n + 1)$. By Fano's inequality, 
$H(X \mid \hat X) \le 1 + \delta \log(n)$.
Thus, 
\begin{align*}
I(X; Q^m, Y^m) \ge I(X; \hat X) = H(X) - H(X \mid  \hat X) \ge (1-\delta-o(1)) \log n.
\end{align*}
On the other hand, we also have 
$$I(X; Q^m, Y^m)  \le \sum_{i\ge 1} \bP[m\ge i] \left(H(Y_i \mid  m\ge i) - H(Y_i \mid  X, Q_i, m\ge i)\right)$$
using the same inequalities in the proof of Lemma~\ref{lemma:groupsort-lowerbound}.

Let $q_i = \bP[Y_i=P \mid  m\ge i]$ and $r_i = \bP[Y_i=S \mid  m\ge i]$. Because $\cA$ exits as soon as it sees $P$ or $S$ we have
\begin{align*}
\sum_{i\ge 1} \bP[m\ge i] (q_i+r_i) \le 1.
\end{align*}

Then
\begin{align*}
H(Y_i \mid  m\ge i) \le h\left(q_i,r_i, \frac{1-q_i-r_i}2, \frac{1-q_i-r_i}2\right)
\le 
h\left(\frac{q_i+r_i}{2},\frac{q_i+r_i}{2}, \frac{1-q_i-r_i}2, \frac{1-q_i-r_i}2\right)
\end{align*}
and
\begin{align*}
H(Y_i \mid  X, Q_i, m\ge i) = 0 \cdot (q_i+r_i) + h(p) (1-q_i-r_i).
\end{align*}

Note that $H(Y_i \mid  m\ge i) \le 2$ trivially, so $I(X; Q^m, Y^m) \le 2  \sum_{i \ge 1} \bP[m \ge i] = 2 \bE[m]$. Combining with $I(X; Q^m, Y^m) \ge (1-\delta-o(1)) \log n$, we have $\bE[m] = \Omega(\log n)$ (as we can assume $\delta \le 1-\Omega(1)$).

Note that $g(x) := h(\frac{x}{2},\frac{x}{2},\frac{1-x}2,\frac{1-x}2)$ is concave in $x$, so
\begin{align*}
\sum_{i\ge 1} \bP[m\ge i] H(Y_i \mid  m\ge i) &\le \left(\sum_{i\ge 1} \bP[m\ge i]\right) g\left(\frac{\sum_{i\ge 1} \bP[m\ge i] (q_i+r_i)/2}{\sum_{i\ge 1} \bP[m\ge i]}\right) \tag{Jensen's inequality}\\
& \le \bE[m] \cdot g\left(\min\left\{\frac 12, \frac{1/2}{\bE[m]}\right\}\right)  \tag{$g$ is increasing on $[0, \frac 12]$ and is maximum at $\frac 12$}\\
& = (1+o(1))  \cdot \bE[m].\tag{$\frac{1/2}{\bE[m]} = o(1)$}
\end{align*}

On the other hand,
\begin{align*}
\sum_{i\ge 1} \bP[m\ge i] H(Y_i \mid  X, Q_i, m\ge i) &= \sum_{i\ge 1} \bP[m\ge i] h(p) (1-q_i-r_i) \\
&= h(p) \left(\sum_{i\ge 1} \bP[m\ge i] - \sum_{i\ge 1} \bP[m\ge i] (q_i+r_i)\right) \\
&\ge h(p) (\bE[m] - 1) \\
&= (1-o(1)) h(p) \bE[m].
\end{align*}

Therefore,
\begin{align*}
I(X; Q^m, Y^m) &\le (1+o(1)) \bE[m] - (1-o(1)) h(p) \bE[m] \\
&\le (1+o(1)) I(p) \bE[m].
\end{align*}

Combining with $I(X; Q^m, Y^m) \ge (1-\delta-o(1)) \log n$, we get $\bE[m] \ge (1-\delta-o(1))\frac{\log n}{I(p)}$.
\end{proof}

\begin{lemma} \label{lemma:binarysearch-neighbor-lowerbound}

Fix any (deterministic or randomized) algorithm for \NoisyBinarySearch{n} with error probability $\delta\le 1/\log n$.
Let $L$ be the input list and $x$ be the input element of \NoisyBinarySearch{n}.  Let $U^{\approx}$ denote the number of noisy comparisons made between $x$ and it the predecessor and successor. 
Then $\bE[U^{\approx}] \ge (2-o(1)) \frac{\log \frac{1}{\delta}}{(1-2p) \log \frac {1-p}p}$, even if the position of the element to search for is uniformly random. 
\end{lemma}
\begin{proof}

Let the input list contain elements $y_1, \ldots, y_n$ in this order. Let $k$ be the index of the predecessor of $x$.
Define $W_i$ to be the number of queries returning $x < y_i$, minus the number of queries returning $x > y_i$. Note that $W_i$ is a random variable depending on the transcript. We will first show that $\bE[U^{\approx} \mid k \not \in\{0, n\}] \ge (2-o(1)) \frac{\log \frac{1}{\delta}}{(1-2p) \log \frac {1-p}p}$. Notice that the error probability of the algorithm conditioned on $k \not \in\{0, n\}$ is at most $\frac{\delta}{\Pr[k \not \in\{0, n\}]} \le 1.1\delta$ for sufficiently large $n$.

Define $W = W_k + W_{k+1}$. 
Then $\bE[W] = (1-2p)\bE[U^{\approx}]$.
Let us consider the posterior probabilities $q$ of the distributions of $k$ given a fixed transcript.
That is, $$q_i := \bP[k=i \mid  Q^m, Y^m].$$
Then if $k \not \in \{0, n\}$, similarly to the proof of Lemma~\ref{lemma:noisysort-samegroup-lowerbound}, we have
$$q_{k-1} = q_k \cdot \left(\frac{1-p}{p} \right)^{-W_k} \text{\quad and\quad } q_{k+1} = q_k \cdot \left(\frac{1-p}{p} \right)^{-W_{k+1}}$$

Thus, by Jensen's inequality, 
\begin{align*}
q_{k-1}+q_{k+1} \ge q_k \left(\left(\frac{1-p}{p} \right)^{-W_k} + \left(\frac{1-p}{p} \right)^{-W_{k+1}}\right) \ge 2q_k \left(\frac{1-p}p\right)^{-W/2}.
\end{align*}

Suppose for the sake of contradiction that 
\begin{align*}
\bE[U^{\approx} \mid k \not \in \{0, 1\}] \le (2-\sigma)\frac{\log \frac{1}{\delta}}{(1-2p) \log\frac{1-p}p}
\end{align*}
for some constant $\sigma > 0$. 
Then by Lemma~\ref{lemma:noisysort-samegroup-lowerbound-technical} (taking $U^{\approx}$ as $\tau$, $W$ as $X_\tau$, $m$ as $\frac{(2-\sigma/4) \log \frac{1}{\delta}}{\log\frac{1-p}p}$, $\delta$ as $1-\frac{2-\sigma}{2-\sigma/4} \ge \sigma/4$), with probability $\ge \sigma/8-o_{1/\delta}(1)$, we have
\begin{align*}
W \le \frac{(2-\sigma/4) \log \frac{1}{\delta}}{\log\frac{1-p}p},
\end{align*}
and then
\begin{align*}
2\left(\frac{1-p}p\right)^{-W/2} \ge 2\delta^{1-\sigma/8}.
\end{align*}

Because the error probability is $\le 1.1\delta$, with probability at least $1-\frac{1.1\delta}{10\delta/\sigma} \ge 1-\frac{\sigma}{9}$, we have $1-q_k \le 10\delta/\sigma$, which leads to $q_{k-1}+q_{k+1} \le 10\delta/\sigma$. Therefore, $$2\left(\frac{1-p}p\right)^{-W/2} \le \frac{q_{k-1}+q_{k+1}}{q_k} \le \frac{10\delta/\sigma}{1-10\delta/\sigma} < 20\delta/\sigma.$$
(In the last step we use $\delta \le 1/\log n$ and thus $1-10\delta/\sigma > \frac{1}{2}$ for sufficiently large $n$.)
This is a contradiction because $2\delta^{1-\sigma/8} \ge 20\delta/\sigma$ and $\frac{\sigma}{8} + \left(1-\frac{\sigma}{9}\right) - o_{1/\delta}(1) > 1$ when $\delta \le 1/\log n$ and $n$ is large enough.

Thus, 
\begin{align*}
\bE[U^{\approx} \mid k \not \in \{0, n\}] \ge (2-o(1))\frac{\log \frac{1}{\delta}}{(1-2p) \log\frac{1-p}p}.
\end{align*}
Consequently, 
\begin{align*}
\bE[U^{\approx}] \ge \bE[U^{\approx} \mid k \not \in \{0, n\}] \cdot \Pr[k \not \in \{0, n\}] \ge (2-o(1))\frac{\log \frac{1}{\delta}}{(1-2p) \log\frac{1-p}p}.
\end{align*}
\end{proof}

Now we are ready to prove Theorem~\ref{thm:binarysearch-lowerbound}.
\begin{proof}[Proof of Theorem~\ref{thm:binarysearch-lowerbound}]
If $\delta \ge 1/\log n$, then by Lemma \ref{lemma:binary-search-to-oracle-binary-search},  \ref{lemma:oracle-binary-search-lowerbound},
\begin{align*}
    \bE[U] &\ge \bE[U^{\approx}] \\
    &\ge (1-\delta-o(1)) \frac{\log n}{I(p)}\\
    &\ge (1-o(1)) \left((1-\delta)\frac{\log n}{I(p)} + \frac{2 \log \frac 1\delta}{(1-2p)\log\frac {1-p}p}\right).
\end{align*}

If $\delta \le 1/\log n$, then by combining Lemma \ref{lemma:binary-search-to-oracle-binary-search},  \ref{lemma:oracle-binary-search-lowerbound}, \ref{lemma:binarysearch-neighbor-lowerbound}, for any algorithm, the number of queries $U$ satisfies
\begin{align*}
    \bE[U] &\ge \bE[U^{\approx}] + \bE[U^{\not \approx}]\\
    &\ge (1-\delta-o(1)) \frac{\log n}{I(p)} + (2-o(1)) \frac{\log \frac 1\delta}{(1-2p) \log \frac{1-p}p} \\
    &\ge (1-o(1)) \left((1-\delta)\frac{\log n}{I(p)} + \frac{2 \log \frac 1\delta}{(1-2p)\log\frac {1-p}p}\right).
\end{align*}
\end{proof}

\subsection{Upper Bound}

In this section we prove our upper bound for Noisy Binary Search.
\SearchUpperBound*

\begin{proof}

First, if $\delta > \frac{1}{\log n}$, we output an arbitrary answer and halt with probability $\delta - \frac{1}{\log n}$.
Otherwise, we call Theorem~\ref{thm:BHbinary} with error probability  $\frac{1}{\log n}$. By union bound, the overall error probability is at most $\delta$. The expected number of noisy comparisons is $$\left(1-\delta +  \frac{1}{\log n} \right) \cdot (1+o(1)) \left( \frac{\log n + O(\log \log n)}{I(p)}\right)=(1+o(1))(1-\delta) \frac{\log n}{I(p)}$$
as required. In the following, we assume $\delta \le \frac{1}{\log n}$. We also assume $\log n \ge 4$ for simplicity, as we could pad dummy elements.

Let  $x$ be the element for which we need to find the predecessor. 
First, we use Theorem~\ref{thm:BHbinary} with error probability $\frac{1}{\log n}$ to find a candidate predecessor, which takes $(1+o(1))\left(\frac{\log n + O(\log \log n)}{I(p)} \right) = (1+o(1))\cdot \frac{\log n}{I(p)}$ time. Then we use Lemma~\ref{lem:repeated_query} to compare $x$ with this candidate predecessor $l$ and with the next element of this candidate predecessor $r$, with error probability $\delta/4$. This takes $(2+o_{4/\delta}(1))\cdot \frac{ \log (4/\delta)}{(1-2p) \log((1-p)/p)} = (2+o(1))\cdot \frac{ \log (1/\delta)}{(1-2p) \log((1-p)/p)}$ time.
If the comparison results are $l < x$ and $x < r$, we return $l$ as the predecessor of $x$, otherwise, we restart. See Algorithm~\ref{algo:noisy_binary_search}. 

\begin{algorithm}[h]
\caption{} \label{algo:noisy_binary_search}
\begin{algorithmic}[1]
\Procedure{NoisyBinarySearch}{$n, A, x, \delta$}
\While{true}
\State Use Theorem~\ref{thm:BHbinary} with error probability $\frac{1}{\log n}$ to search the predecessor $l$ of $x$ in $A \cup \{-\infty\}$.
\State $r \gets $ next element of $l$ in $A \cup \{\infty\}$.
\If{\textsc{LessThan}($l, x,  \delta/4$) and \textsc{LessThan}($x, r,  \delta/4$)}
\Return $l$.
\EndIf
\EndWhile
\EndProcedure
\end{algorithmic}
\end{algorithm}

Let $\bP[\text{restart}]$ be the probability that we restart a new iteration of the while loop. Clearly, if the call to Theorem~\ref{thm:BHbinary} and both calls to $\textsc{LessThan}$ are correct, then we will not restart. Thus, $\bP[\text{restart}] \le \frac{1}{\log n} + \frac{\delta}{2} \le \frac{2}{\log n}$. Note that as $\log n \ge 4$, $\bP[\text{restart}] \le \frac{1}{2}$. 

Let $P_e$ be the probability that this algorithm returns the wrong predecessor. In each iteration of the while loop, the algorithm returns the wrong answer only if at least one of the two calls to $\textsc{LessThan}$ are incorrect. Therefore, 
$P_e \le \frac{\delta}{2} + \bP[\text{restart}] P_e$, so $P_e \le \frac{\delta / 2}{1 - \bP[\text{restart}]} \le \delta$.

Let $\bE[Q]$ be the expected number of queries that this algorithm makes. Since the call to Theorem~\ref{thm:BHbinary} uses $(1+o(1)) \left(\frac{\log n + O(\log \log n)}{I(p)}\right) \le (1+o(1)) \frac{\log n}{I(p)}$ noisy comparisons in expectation, and each call to $\textsc{LessThan}$ uses $(1+o_{1/\delta}(1)) \frac{\log (4/\delta)}{(1-2p) \log \frac{1-p}{p}} \le (1+o(1)) \frac{\log(1/\delta)}{(1-2p) \log \frac{1-p}{p}}$ noisy comparisons in expectation (as $\delta \le \frac{1}{\log n}$), we get that 
$$\bE[Q] \le (1+o(1))\left( \frac{\log n}{I(p)} + \frac{ 2\log (1/\delta)}{(1-2p) \log \frac{1-p}{p}}\right) + \bE[Q] \cdot \bP[\text{restart}].$$ 
Therefore, 
\begin{align*}
\bE[Q] &\le (1+o(1)) \cdot \frac{1}{1-\frac{2}{\log n}} \cdot \left( \frac{\log n}{I(p)} + \frac{ 2\log (1/\delta)}{(1-2p) \log \frac{1-p}{p}}\right)\\
&\le (1+o(1)) \left( \frac{\log n}{I(p)} + \frac{2 \log (1/\delta)}{(1-2p) \log \frac{1-p}{p}}\right),
\end{align*}
as desired.
\end{proof}

\section{Discussions}
In this section we discuss a few possible extensions of our results.

\paragraph{Varying $p$.}
In our main results for Noisy Sorting, we have assumed that $p$ remains constant as $n\to \infty$. This assumption is necessary in several places of our proofs as discussed below and we leave it as an open problem to drop this assumption. 

For $p\to \frac 12$ as $n\to \infty$, Theorem~\ref{thm:sort-upperbound} still holds. For the proof, one needs to suitably modify Corollary~\ref{cor:simple_sort}, Lemma~\ref{lem:weak_sort}, Lemma~\ref{lem:variance_bound}, Corollary~\ref{coro:safe_algo}, Lemma~\ref{lem:safe_simple_sort} to handle $p\to \frac 12$ correctly. For example, in our current statement of Corollary~\ref{coro:safe_algo} we ignore constants related to $p$ in the expressions for $\Var[\tau]$. By choosing $m$ in Lemma~\ref{lem:variance_bound} carefully, we can let $\Var[\tau] = (\bE[\tau])^2 \polylog(n)$ even for growing $p$, which suffices for the concentration results.
For the lower bound (Theorem~\ref{thm:sort-lowerbound}), Lemma~\ref{lemma:groupsort-lowerbound} may become problematic. 
It seems possible to modify the proof of Lemma~\ref{lemma:groupsort-lowerbound} to handle the case $I(p) = n^{-o(1)}$.
However, for $I(p) = n^{-\Omega(1)}$ a stronger method is needed.

For $p\to 0$ as $n\to \infty$, the lower bounds hold without any change.
For the algorithms, our derandomization method would result in too many queries when $p=O(1/\log n)$. 

We similarly leave the task of dropping the assumption that $p$ remains a constant in our Noisy Binary Search results as an open problem. 

\paragraph{Non-uniform $p$ and semi-random adversary.}
In our main results we have assumed that the probability of error in the comparisons remains constant.
One possible extension is that in each noisy comparison, the result is flipped with some probability at most $p$, where this probability can be adversarially chosen for each query.
(The algorithm knows $p$ but does not know the flip probability for individual comparisons.)
This model is sometimes called semi-random adversary in literature (e.g., \cite{mao2018minimax}). It is an interesting open problem to generalize our results to this model.

\paragraph{Dependency on $\delta$.}
In our bounds for Noisy Sorting, we simply used $o(1)$ as our error probability. We leave it as an open problem to generalize our bounds to the case where the error probability $\delta$ can also be specified. 

\paragraph{BMS observation.}
One natural extension is to replace $\BSC$ observation noise with a Binary Memoryless Symmetric channel (BMS).
That is, for every noisy comparison, the probability of error is chosen (independently) from a distribution $W$ on $[0,\frac 12]$, and this probability is returned together with the comparison result.
In this case our guess is that
\begin{align*}
(1\pm o(1)) \left( \frac{1}{\bE_{p\sim W}[I(p)]} + \frac 1{\bE_{p\sim W} [(1-2p)\log \frac{1-p}p]}\right) n\log n
\end{align*}
noisy comparisons should be necessary and sufficient.
When $W$ is supported on $[\epsilon,1/2-\epsilon]$ for some constant $\epsilon > 0$ independent of $n$, we expect it to be easier to adapt our proofs to BMS, whereas in the general case, issues described in the Varying $p$ discussion above can occur.

\bibliographystyle{alpha}
\bibliography{ref}
\end{document}